\documentclass[11pt]{article}
\usepackage{pdfpages}
\usepackage{multirow}
\usepackage{array}
\usepackage{graphicx}
\usepackage{subcaption}
\usepackage{latexsym}
\usepackage{pstricks}
\usepackage{pst-node}
\usepackage{pst-tree}
\usepackage{url}
\usepackage{times}
\usepackage{helvet}
\usepackage{courier}
\usepackage{algorithmic}
\usepackage[ruled,vlined,linesnumbered]{algorithm2e}
\usepackage{amsthm}
\usepackage{graphicx,amssymb,amsmath}
\usepackage{mathrsfs}

\usepackage[hypertexnames=false, colorlinks=true, citecolor=blue, linkcolor=red, urlcolor=black]{hyperref}

\usepackage{geometry}
\usepackage{authblk}

\makeatletter
\renewcommand*{\verbatim@font}{\sffamily}

\global\hyphenpenalty=100000 

\title{Chain, Generalization of Covering Code, \\ and Deterministic Algorithm for k-SAT\thanks{A preliminary version of this paper appeared in the proceedings of ICALP 2018 \cite{liu2018ksat}.}}

\author{S. Cliff Liu\\
Princeton University\\
\textsf{sixuel@cs.princeton.edu}
}

\begin{document}

\maketitle


\newcounter{dummy} \numberwithin{dummy}{section}
\newtheorem{lemma}[dummy]{Lemma}
\newtheorem{definition}[dummy]{Definition}
\newtheorem{remark}[dummy]{Remark}
\newtheorem{theorem}[dummy]{Theorem}
\newtheorem{corollary}[dummy]{Corollary}
\newtheorem{claim}[dummy]{Claim}
\newtheorem{observation}[dummy]{Observation}

\renewcommand{\algorithmicrequire}{\textbf{Input:}}
\renewcommand{\algorithmicensure}{\textbf{Output:}}

\begin{abstract}
    We present the current fastest deterministic algorithm for $k$-SAT, improving the upper bound $(2-2/k)^{n + o(n)}$ dues to Moser and Scheder [STOC'11]. The algorithm combines a branching algorithm with the derandomized local search, whose analysis relies on a special sequence of clauses called chain, and a generalization of covering code based on linear programming.

    We also provide a more ingenious branching algorithm for $3$-SAT to establish the upper bound $1.32793^n$, improved from $1.3303^n$.
 \end{abstract}


\section{Introduction}\label{intro}


As the fundamental NP-complete problems, $k$-SAT and especially $3$-SAT have been extensively studied for decades.
Numerous conceptual breakthroughs have been put forward via continued progress of exponential-time algorithms,
including randomized and deterministic ones.

The first provable algorithm for solving $k$-SAT on $n$ variables in less than $2^n$ steps was presented by Monien and Speckenmeyer, using the concept of autark assignment~\cite{monien1985solving}.
Later their bound $1.619^n$ for $3$-SAT was improved to $1.579^n$ and $1.505^n$ respectively~\cite{Schiermeyer1970Solving, DBLP:journals/tcs/Kullmann99}.
These algorithms follow a branching manner, i.e., recursively reducing the formula size by branching and fixing variables deterministically, thus are called \emph{branching algorithms}.

As for randomized algorithms,
two influential ones are PPSZ and Sch\"{o}ning's local search \cite{DBLP:journals/jacm/PaturiPSZ05, schoning1999probabilistic}.
There has been a long line of research improving the bound $(4/3)^n$ of local search for $3$-SAT, including HSSW and local search combined with PPSZ \cite{hofmeister2002probabilistic, iwama2004improved}.
In a breakthrough work, Hertli closed the gap between Unique and General cases for PPSZ \cite{hertli20143}. (By Unique it means the formula has at most one satisfying assignment.)
In a word, considering randomized algorithms, PPSZ for $k$-SAT is currently the fastest, although with one-sided error (see PPSZ in Table~\ref{table_result}).
Unfortunately, PPSZ for the General case seems tough to derandomize due to the excessive usage of random bits \cite{DBLP:conf/sat/Rolf05, DBLP:journals/corr/abs-2001-06536}.

In contrast to the hardness in derandomizing PPSZ, local search can be derandomized using the so-called covering code \cite{dantsin2002deterministic}.
Subsequent deterministic algorithms focused on boosting local search for $3$-SAT to the bounds $1.473^n$ and $1.465^n$ \cite{DBLP:journals/tcs/BrueggemannK04, DBLP:conf/latin/Scheder08}.
In 2011, Moser and Scheder fully derandomized Sch\"{o}ning's local search with another covering code for the choice of flipping variables within the unsatisfied clauses, which was immediately improved by derandomizing HSSW for $3$-SAT, leading to the current best upper bounds for $k$-SAT (see Table~\ref{table_result})
 \cite{moser2011full,DBLP:journals/algorithmica/MakinoTY13}.
Since then, the randomness in Sch\"{o}ning's local search are all replaced by deterministic choices, and the bounds remain untouched.
How to break the barrier?

The difficulty arises in both directions.
If attacking this without local search, one has to derandomize PPSZ or propose radically new algorithm.
Else if attacking this from derandomizing local search-based algorithm,
one must greatly reduce the searching space.

Our method is a combination of a branching algorithm and the derandomized local search.
As we mentioned in the second paragraph of this paper, branching algorithm is intrinsically deterministic, therefore it remains to leverage the upper bounds for both of them by some tradeoff.
The tradeoff we found is the weighted size of a carefully chosen set of chains,
where a chain is a sequence of clauses sharing variable with the clauses next to them only,
such that a branching algorithm either solves the formula within desired time or returns a large enough set of chains.
The algorithm is based on the study of autark assignment from
\cite{monien1985solving} with further refinement,
whose output
can be regarded as a generalization of maximal independent clause set from HSSW \cite{hofmeister2002probabilistic}, which reduces the $k$-CNF to a $(k-1)$-CNF.
\footnote{We refer the reader to Chapter 11 in \cite{DBLP:series/faia/BuningK09} for a survey of autark assignment.}
The searching space equipped with chains is rather different from those in previous derandomizations \cite{dantsin2002deterministic, DBLP:journals/algorithmica/MakinoTY13, moser2011full}: it is a Cartesian product of finite number of non-uniform spaces.
Using linear programming, we prove that such space can be perfectly covered, and searched by derandomized local search within aimed time.
Additionally, unlike the numerical upper bound in HSSW \cite{hofmeister2002probabilistic}, we give the closed form.

The rest of the paper is organized as follows.
In \S{\ref{pre}} we give basic notations, definitions related to chain and algorithmic framework.
We show how to generalize covering code to cover any space equipped with chains in \S{\ref{GCC_section}}.
Then we use such code in derandomized local search in \S{\ref{DLS_section}}.
In \S{\ref{UBK_section}}, we prove upper bound for $k$-SAT.
A more ingenious branching algorithm for $3$-SAT in \S{\ref{UB3_section}} is presented.
Some upper bound results are highlighted in Table~\ref{table_result}, with main results formally stated in Theorem~\ref{main_k_sat_general_form} of \S{\ref{UBK_section}} and Theorem~\ref{main2} of \S{\ref{UB3_section}}.
We conclude this paper in \S{\ref{conclusion}} with some discussions.

\begin{table}
\centering
\caption{The rounded up base $c$ in the upper bound $c^n$ of our deterministic algorithm for $k$-SAT and the corresponding upper bound in previous results \cite{DBLP:journals/algorithmica/MakinoTY13, moser2011full, dantsin2002deterministic} as well as in the currently fastest randomized algorithm \cite{DBLP:journals/jacm/PaturiPSZ05, hertli20143}.
}\label{table_result}
\begin{tabular}{|c|c|c|c|c|c|}\hline
$k$ & Our Result & Makino et al. & Moser\&Scheder & Dantsin et al. & PPSZ(randomized) \\ \hline \hline
3 & \textbf{1.32793} & 1.3303 & 1.33334 & 1.5 & 1.30704 \\ \hline

4 & \textbf{1.49857} & - & 1.50001 & 1.6 & 1.46899 \\ \hline

5 & \textbf{1.59946} & - & 1.60001 & 1.66667 & 1.56943 \\ \hline

6 & \textbf{1.66646} & - & 1.66667 & 1.71429 & 1.63788 \\ \hline

\end{tabular}
\end{table}

\section{Preliminaries}\label{pre}

\subsection{Notations}

We study formulae in Conjunctive Normal Form (CNF).
Let $V=\{v_i | i \in [n]\}$ be a set of $n$ boolean variables.
For all $i \in [n]$, a literal $l_i$ is either $v_i$ or $\bar{v}_i$.
A clause $C$ is a disjunction of literals and a CNF $F$ is a conjunction of clauses.
A $k$-clause is a clause that consists of exactly $k$ literals, and an $\le k$-clause consists of at most $k$ literals.
If every clause in $F$ is $\le k$-clause, then $F$ is a $k$-CNF.

An \emph{assignment} is a function $\alpha: V \mapsto \{0, 1\}$ that maps each $v \in V$ to truth value $\{0,1\}$.
A \emph{partial assignment} is the function restricted on $V' \subseteq V$. 
We use $F|\alpha(V')$ to denote the formula derived by fixing the values of variables in $V'$ according to partial assignment $\alpha(V')$.
A clause $C$ is said to be \emph{satisfied} by $\alpha$ if $\alpha$ assigns at least one literal in $C$ to $1$.
$F$ is \emph{satisfiable} iff there exists an $\alpha$ satisfying all clauses in $F$, and we call such $\alpha$ a \emph{satisfying assignment} of $F$.
The $k$-SAT problem asks to find a satisfying assignment of a given $k$-CNF $F$ or to prove its non-existence if $F$ is unsatisfiable.

Let $X$ be a literal or a clause or a collection of either of them, we use $V(X)$ to denote the set of all the variables appear in $X$.
We say that $X$ and $X'$ are \emph{independent} if $V(X) \cap V(X') = \emptyset$, or $X$ \emph{overlaps} with $X'$ if otherwise.

A \emph{word} of length $n$ is a vector from $\{0, 1\}^n$.
The \emph{Hamming space} $H \subseteq \{0, 1\}^n$ is a set of words.
Given two words $\alpha_1, \alpha_2 \in H$, the \emph{Hamming distance} $d(\alpha_1, \alpha_2) = \|\alpha_1 - \alpha_2 \|_1$ is the number of bits $\alpha_1$ and $\alpha_2$ disagree.
The reason of using $\alpha$ for word as same as for assignment is straightforward: Giving each variable an index $i \in [n]$, a word of length $n$ naturally corresponds to an assignment, which will be used interchangeably.

Throughout the paper, $n$ always denotes the number of variables in the formula and will be omitted if the context is clear.
We use $O^*(f(n)) = \text{poly}(n) \cdot f(n)$ to suppress polynomial factors, and use $\mathcal{O}(f(n)) = 2^{o(n)} \cdot f(n)$ to suppress sub-exponential factors.

\subsection{Preliminaries for Chains}\label{PC_section}

In this subsection, we propose our central concepts, which are the basis of our analysis.


\begin{definition}\label{chain_def}
    Given integers $k \ge 3$ and $\tau \ge 1$,
    a $\tau$-\emph{chain} $\mathcal{S}^{(k)}$ is a sequence of $\tau$ $k$-clauses $\langle C_1, \dots, C_{\tau} \rangle$ satisfies that $\forall i, j \in [\tau]$, $V(C_i) \cap V(C_j) = \emptyset$ iff $|i - j | > 1$.
\end{definition}

If the context is clear, we will use $\mathcal{S}$, $\tau$-chain or simply chain for short.

\begin{definition}\label{instance_def}
    A set of chains $\mathcal{I}$ is called an \emph{instance} if $\forall \mathcal{S}, \mathcal{S}' \in \mathcal{I}$, $V(\mathcal{S}) \cap V(\mathcal{S}') = \emptyset$ for $\mathcal{S} \neq \mathcal{S}'$.
\end{definition}

In other words, each clause in chain only and must overlap with the clauses next to it (if exist), and chains in an instance are mutually independent.
\begin{definition}
    Given chain $\mathcal{S}$, define the \emph{solution space} of $\mathcal{S}$ as $A \subseteq \{0, 1\}^{|V(\mathcal{S})|}$ such that partial assignment $\alpha$ on $V(\mathcal{S})$ satisfies all clauses in $\mathcal{S}$ iff $\alpha(V(\mathcal{S})) \in A$.
    \footnote{This essentially defines the set of all satisfying assignments for a chain. As a simple example in $3$-CNF, $1$-chain $\langle x_1 \vee x_2 \vee x_3 \rangle$ has solution space $A = \{0,1\}^3 \backslash 0^3$.}
\end{definition}

We define vital algebraic property of chain, which will play a key role in the construction of our generalized covering code.

\begin{definition}\label{charactoristic_def}
    Let $A$ be the solution space of chain $\mathcal{S}^{(k)}$, define $\lambda \in \mathbb{R}$ and $\pi: A \mapsto [0, 1]$ as the \emph{characteristic value} and \emph{characteristic distribution} of $\mathcal{S}^{(k)}$ respectively, where $\lambda$ and $\pi$ are feasible solution to the following linear programming $\text{LP}_A$:
    \begin{align*}
        & \sum_{a \in A} \pi(a) = 1 \\
        & \lambda = \sum_{a \in A} \left( \pi(a) \cdot (\frac{1}{k-1}) ^ {d(a, a^*)} \right) & \forall a^* \in A \\
        & \pi(a) \ge 0 & \forall a \in A
    \end{align*}
\end{definition}
\begin{remark}
    The variables in $\text{LP}_A$ are $\lambda$ and $\pi(a)~(\forall a \in A)$.
    There are $|A| + 1$ variables and $|A| + 1$ equality constraints in $\text{LP}_A$.
    One can work out the determinant of the coefficient matrix to see it has full rank,
    so the solution is unique if feasible. Specifically, $\lambda \in (0,1)$.

\end{remark}


\subsection{Algorithmic Framework}\label{AF_section}

Our algorithm (Algorithm~\ref{Framework}) is a combination of a branching algorithm called \textsf{BR}, and a derandomized local search called \textsf{DLS}.
\textsf{BR} either solves $F$ or provides a large enough instance to \textsf{DLS} for further use, which essentially reduces the Hamming space exponentially.

\begin{algorithm}
\caption{Algorithmic Framework}
\label{Framework}
\begin{algorithmic}[1]
\REQUIRE $k$-CNF $F$
\ENSURE a satisfying assignment or \verb"Unsatisfiable"
\STATE \textsf{BR}$(F)$ either solves $F$ or returns an instance $\mathcal{I}$
\IF {$F$ is not solved}
    \STATE \textsf{DLS}$(F,\mathcal{I})$
\ENDIF
\end{algorithmic}
\end{algorithm}


\section{Generalization of Covering Code}\label{GCC_section}

First of all, we introduce the covering code, then show how to generalize it for the purpose of our derandomized local search.

\subsection{Preliminaries for Covering Code}

The \emph{Hamming ball} of \emph{radius} $r$ and \emph{center} $\alpha$ $B_\alpha(r) = \{ \alpha' | d(\alpha, \alpha') \le r \}$ is the set of all words with Hamming distance at most $r$ from $\alpha$.
A \emph{covering code} of \emph{radius} $r$ for Hamming space $H$  is a set of words $C(r) \subseteq H$ satisfies $\forall \alpha' \in H, \exists \alpha \in C(r)$, such that $d(\alpha, \alpha') \le r$, i.e., $H \subseteq \bigcup_{\alpha \in C(r)} B_\alpha(r)$, and we say $C(r)$ \emph{covers} $H$.

Let $\ell$ be a non-negative integer and set $[\ell]^* = [\ell] \cup \{0\}$, a set of covering codes $\{C(r) | r \in [\ell]^*\}$ is an \emph{$\ell$-covering code} for $H$ if $\forall r \in [\ell]^*, C(r) \subseteq H$ and $H \subseteq \bigcup_{r \in [\ell]^*} \bigcup_{\alpha \in C(r)} B_\alpha(r)$, i.e., $\{C(r) | r \in [\ell]^*\}$ covers $H$.

The following lemma gives the construction time and size of covering codes for the uniform Hamming spaces $\{0, 1\}^n$.

\begin{lemma}[\cite{dantsin2002deterministic}]\label{cover_01_space}
    Given $\rho \in (0, \frac{1}{2})$, there exists a covering code $C(\rho n)$ for Hamming space $\{0, 1\}^n$, such that $|C(\rho n)| \le O^*(2^{(1 - h(\rho)) n})$ and $C(\rho n)$ can be deterministically constructed in time $O^*(2^{(1 - h(\rho)) n})$, where $h(\rho)=-\rho\log{\rho}-(1-\rho)\log{(1-\rho)}$ is the \emph{binary entropy function}.
\end{lemma}


\subsection{Generalized Covering Code}\label{GCC_subsection}

In this subsection we introduce our generalized covering code, including its size and construction time.

First of all we take a detour to define the \emph{Cartesian product} of $\sigma$ sets of words as $ X_1 \times \dots \times X_{\sigma} = \prod_{i \in [\sigma]} X_i = \{ \uplus_{i \in [\sigma]} \alpha_i | \forall i \in [\sigma], \alpha_i \in X_i \}$, where $\uplus_{i \in [\sigma]} \alpha_i$ is the concatenation from $\alpha_1$ to $\alpha_{\sigma}$.
Then we claim that the Cartesian product of covering codes is also a good covering code for the Cartesian product of the Hamming spaces they covered separately.

\begin{lemma}\label{product_all}
    Given integer $\chi > 1$, for each $i \in [\chi]$, let $H_i$ be a Hamming space and $C_i(r_i)$ be a covering code for $H_i$.
    If $C_i(r_i)$ can be deterministically constructed in time $O^*(f_i(n))$ and $|C_i(r_i)| \le O^*(g_i(n))$ for all $i \in [\chi]$, then there exists covering code $\mathfrak{C}$ of radius $\sum_{i \in [\chi]} r_i$ for Hamming space $\prod_{i \in [\chi]} H_i$ such that $\mathfrak{C}$ can be deterministically constructed in time $O^*(\sum_{i \in [\chi]} f_i(n) + \prod_{i \in [\chi]} g_i(n))$ and $|\mathfrak{C}| \le O^*(\prod_{i \in [\chi]} g_i(n))$.
\end{lemma}
\begin{proof}
    We prove that covering code $\prod_{i \in [\chi]} C_i(r_i)$ can be such $\mathfrak{C}$.
    For any $\alpha' \in \prod_{i \in [\chi]} H_i$, one can write $\alpha' = \uplus_{i \in [\chi]} {\alpha_i}'$ where ${\alpha_i}' \in H_i$.
    Then by definition, $\exists \alpha_i \in C_i(r_i)$, such that $d(\alpha_i, {\alpha_i}') \le r_i$.
    Now let $\alpha = \uplus_{i \in [\chi]} \alpha_i$, we have that $d(\alpha, \alpha') = \sum_{i \in [\chi]} d(\alpha_i, {\alpha_i}') \le \sum_{i \in [\chi]} r_i$.

    To construct $\mathfrak{C}$, we first construct all $C_i(r_i)$ in time $O^*(\sum_{i \in [\chi]} f_i(n))$, then concatenate every $\alpha_1$ to every $\alpha_{\chi}$, which can be done in time $O^*(\prod_{i \in [\chi]} |C_i(r_i)|)$. So the total construction takes time $O^*(\sum_{i \in [\chi]} f_i(n) + \prod_{i \in [\chi]} g_i(n))$ and is deterministic.
    Obviously $|\mathfrak{C}| \le O^*(\prod_{i \in [\chi]} g_i(n))$.
    Therefore we proved the lemma.
\end{proof}

Our result on generalized covering code is given below.

\begin{lemma}\label{general_space}
    Let $A$ be the solution space of chain $\mathcal{S}^{(k)}$ whose characteristic value is $\lambda$, for any $\nu = \Theta(n)$,
    there exists an $\ell$-covering code $\{C(r) | r \in [\ell]^*\}$ for Hamming space $H = {A}^{\nu}$ where $\ell = \lfloor - \nu \log_{k-1}{\lambda} + 2 \rfloor$, such that $|C(r)| \le O^*({\lambda}^{- \nu} / (k-1)^r)$ and $C(r)$ can be deterministically constructed in time $O^*({\lambda}^{- \nu} / (k-1)^r)$, for all $r \in [\ell]^*$.
\end{lemma}
\begin{proof}
    First of all, we show the existence of such $\ell$-covering code by a probabilistic argument.
    For each $r \in [l]^*$, let $s(r) = \lceil -3 \log_{k-1}\lambda \cdot \ln|A| \cdot {\nu}^2 {\lambda}^{- \nu} / (k-1)^r \rceil$.
    We build $C(r)$ from $\emptyset$ by repeating the following for $s(r)$ times independently: choose $\nu$ words $a_j~(j \in [\nu])$ independently from $A$ according to distribution $\pi$ and concatenating them to get a word $\alpha \in A^{\nu}$, then add $\alpha$ to $C(r)$ with replacement, where $\pi$ is the characteristic distribution (Definition~\ref{charactoristic_def}).
    Clearly, $|C(r)| \le s(r) = O^*(\lambda^{-\nu} / (k-1)^r)$.
    We have that for arbitrary fixed $\alpha^* = \uplus_{j \in [\nu]} a_j \in A^{\nu}$, the following must hold:
    \begin{align}
        & \sum_{r} \left( \Pr[d(\alpha, \alpha^*) = r] \cdot (\frac{1}{k-1}) ^ r \right) \notag \\
        = & \sum_r \sum_{\sum_{j \in [\nu]} r_j = r} \left( \Pr[\textit{for all}~j \in [\nu],d(a_j, {a_j}^*) = r_j] \cdot (\frac{1}{k-1})^{\sum_{j \in [\nu]} r_j} \right) \notag \\
        = & \prod_{j \in [\nu]} \sum_{r_j} \left( \Pr[d(a_j, {a_j}^*) = r_j] \cdot (\frac{1}{k-1})^{r_j} \right) \notag \\
        = & \left( \sum_{a \in A} \pi(a) \cdot (\frac{1}{k-1}) ^ {d(a, a^*)} \right)^{\nu} \notag \\
         = & \lambda^{\nu} . \label{expectation1}
    \end{align}
    The third line follows independence and the last line follows from the definition of $\pi$.
    We now show that $\{C(r) | r \in [\ell]^*\}$ covers $H$ with positive probability.
    Rewrite (\ref{expectation1}) as:
    \begin{align}
        \lambda^{\nu} 
        & = \sum_{r \le \ell} \left( \Pr[d(\alpha, \alpha^*) = r] \cdot (\frac{1}{k-1}) ^ r \right) + \mathbf{1}_{r > \ell} \cdot \sum_{r > \ell} \left( \Pr[d(\alpha, \alpha^*) = r] \cdot (\frac{1}{k-1}) ^ r \right) \notag \\
        & \le \sum_{r \le \ell} \left( \Pr[d(\alpha, \alpha^*) = r] \cdot (\frac{1}{k-1}) ^ r \right) + (\frac{1}{k-1})^{\ell} \notag \\
        & \le \sum_{r \le \ell} \left( \Pr[d(\alpha, \alpha^*) = r] \cdot (\frac{1}{k-1}) ^ r \right) + \lambda^{\nu} / (k-1) \notag .
    \end{align}
    The last inequality follows by $\ell \ge -\nu \log_{k-1}(\lambda) + 1$. Thus we have:
    \begin{equation*}
        \sum_{r \le \ell} \left( \Pr[d(\alpha, \alpha^*) = r] \cdot (\frac{1}{k-1}) ^ r \right) \ge \frac{k-2}{k-1} \lambda^{\nu}.
    \end{equation*}
    Then there must exist $r^* \in [\ell]^*$ such that $\Pr[d(\alpha, \alpha^*) = r^*] \ge \lambda^{\nu} (k-2) (k-1)^{r^*-1} / (\ell + 1)$. Using this as the lower bound for $\Pr[d(\alpha, \alpha^*) \le r^*]$, we obtain:
    \begin{align*}
        \Pr[\alpha^* \notin \bigcup_{r \in [\ell]^*} \bigcup_{\alpha \in C(r)} B_\alpha(r)] & \le \Pr[\alpha^* \notin \bigcup_{\alpha \in C(r^*)} B_\alpha(r^*)] \\
        & \le (1 - \Pr[d(\alpha, \alpha^*) \le r^*]) ^ {s(r^*)} \\
        &\le (1 - \lambda^{\nu} (k-2) (k-1)^{r^*-1} / (\ell + 1)) ^ {s(r^*)} \\
        &\le \exp(- \lambda^{\nu} (k-2) (k-1)^{r^*-1} / (\ell + 1)  \cdot s(r^*) ) \\
        &\le {|A|}^{-2\nu} .
    \end{align*}
    The last inequality follows from $s(r) \ge -3 \log_{k-1}\lambda \cdot \ln|A| \cdot {\nu}^2 {\lambda}^{- \nu} / (k-1)^r$ and $\ell \le - \nu \log_{k-1}{\lambda} + 2 $.
    There are $|A|^{\nu}$ words in $H$, so the probability that some $\alpha^* \in H$ is not covered by any $C(r)$ is upper bounded by $|A|^{\nu} \cdot {|A|}^{-2\nu} = {|A|}^{-\nu} = 2^{-\Theta(n)} < 1$. As a result, the $\ell$-covering code in Lemma~\ref{general_space} exists.

    The argument for size and construction is as same as in \cite{DBLP:journals/algorithmica/MakinoTY13}.
    W.l.o.g., let $d \ge 2$ be a constant divisor of $\nu$. By partitioning $H$ into $d$ blocks and applying the approximation algorithm for the set covering problem in \cite{dantsin2002deterministic}, we have that an $(\ell / d)$-covering code for each block can be deterministically constructed in time $O^*({|A|}^{3 \nu / d})$ and $|C(r)| \le O^*(\lambda^{-\nu / d} / (k-1)^r)$ for each $r \in [\ell / d] ^*$, because we can explicitly calculate $r^*$ for each word to cover.
    To get an $\ell$-covering code, note that any $r \in [\ell]^*$ can be written as $r = \sum_{j \in [d]} r_j$ where $r_j \in [\ell / d]^*$, thus $C(r)$ can be constructed by taking Cartesian product of $d$ covering codes $C(r_j)~(j \in [d])$. So by Lemma~\ref{product_all}, the construction time for $C(r)$ is:
    \begin{equation*}
        \sum_{\sum_{j \in [d]} r_j = r} \left( O^*(\sum_{j \in [d]} {|A|}^{3 \nu / d} + \prod_{j \in [d]} \lambda^{-\nu / d} / (k-1)^{r_j}) \right) = O^*(\lambda^{-\nu} / (k-1)^r).
    \end{equation*}
    The equality follows by taking large enough $d$ and observing that there are $O(r^d)$ ways to partition $r$ into $d$ positive integers.
    Also by Lemma~\ref{product_all}, the size of the concatenated covering code is upper bounded by its construction time, which is $|C(r)| \le O^*(\lambda^{-\nu} / (k-1)^r)$.
    Therefore we proved this lemma.
\end{proof}

\section{Derandomized Local Search}\label{DLS_section}

In this section,
we present our derandomized local search (\textsf{DLS}),
see Algorithm~\ref{dls_alg}.

\begin{algorithm}
\caption{Derandomized Local Search: \textsf{DLS}}\label{dls_alg}
\begin{algorithmic}[1]
\REQUIRE $k$-CNF $F$, instance $\mathcal{I}$
\ENSURE a satisfying assignment or \verb"Unsatisfiable"
\STATE construct covering code $\mathfrak{C}$ for Hamming space $H(F, \mathcal{I})$ (Definition~\ref{Hamming_instance}) \label{line_construct}
\FOR {every word $\alpha \in \mathfrak{C}$}
    \IF {\textsf{searchball-fast$(F, \alpha, r)$} finds a satisfying assignment $\alpha^*$ for $F$} \label{line_searchball}
        \RETURN $\alpha^*$
    \ENDIF
\ENDFOR
\RETURN \verb"Unsatisfiable"
\end{algorithmic}
\end{algorithm}

The algorithm first constructs the generalized covering code and stores it (Line~\ref{line_construct}), then calls \textsf{searchball-fast} (Line~\ref{line_searchball}) to search inside each Hamming ball,
where \textsf{searchball-fast} refers to the same algorithm proposed in \cite{moser2011full}, whose running time is stated in the following lemma.

\begin{lemma}[\cite{moser2011full}]\label{full_ball}
    Given $k$-CNF $F$, if there exists a satisfying assignment $\alpha^*$ for $F$ in $B_\alpha(r)$, then $\alpha^*$ can be found by \textsf{searchball-fast} in time $(k-1)^{r + o(r)}$.
\end{lemma}

Our generalized covering code is able to cover the following Hamming space.

\begin{definition}\label{Hamming_instance}
    Given $k$-CNF $F$ and instance $\mathcal{I}$, the Hamming space for $F$ and $\mathcal{I}$ is defined as $H(F, \mathcal{I}) = H_0 \times \prod_i H_i$, where:
    \begin{itemize}
        \item $H_0 = \{0, 1\}^{n'}$ where $n' = n - |V(\mathcal{I})|$.
        \item $H_i = {A_i}^{\nu_i}$ for all $i$, where $A_i$ is a solution space and $\nu_i = \Theta(n)$ is the number of chains in $\mathcal{I}$ with solution space $A_i$.
            \footnote{As we shall see in \S{\ref{UBK_section}} and \S{\ref{UB3_section}}, there are only finite number of different solution spaces and finite elements in each solution space. Thus for those $\nu_i = o(n)$, we can enumerate all possible combinations of assignments on them and just get a sub-exponential slowdown, i,e., an $\mathcal{O}(1)$ factor in the upper bound. \label{footnote}}
    \end{itemize}
\end{definition}

Apparently all satisfying assignments of $F$ lie in $H(F, \mathcal{I})$, because $\prod_i H_i$ contains all assignments on $V(\mathcal{I})$ which satisfy all clauses in $\mathcal{I}$ and $H_0$ contains all possible assignments of variables outside $\mathcal{I}$. Therefore to solve $F$, it is sufficient to search the entire $H(F, \mathcal{I})$.

\begin{definition}\label{gcc_def}
    Given $\rho \in (0, \frac{1}{2})$ and Hamming space $H(F, \mathcal{I})$ as above,
    for $L \in \mathbb{Z}^*$,
    define covering code $\mathfrak{C}(L)$ for $H(F, \mathcal{I})$ as a set of covering codes $\{C(r) | (r - \rho n') \in [L]^*\}$ satisfies that $C(r) \subseteq H(F, \mathcal{I})$ for all $r$ and $H(F, \mathcal{I}) \subseteq \bigcup_{(r - \rho n') \in [L]^*} \bigcup_{\alpha \in C(r)} B_\alpha(r) $, i.e., $\mathfrak{C}(L)$ covers $H(F, \mathcal{I})$.
\end{definition}
\begin{lemma}\label{all_code}
    Given Hamming space $H(F, \mathcal{I})$ and $A_i, \nu_i$ as above,
    let $L = \sum_i \ell_i$ where $\ell_i = \lfloor -\nu_i \log\lambda_i + 2 \rfloor$ and $\lambda_i$ is the characteristic value of chain with solution space $A_i$.
    Given $\rho \in (0, \frac{1}{2})$, covering code $\mathfrak{C}(L) = \{C(r) | (r - \rho n') \in [L]^*\}$ for $H(F, \mathcal{I})$ can be deterministically constructed in time $O^*(2^{(1 - h(\rho)) n'} \prod_i {\lambda_i}^{-\nu_i} )$ and $|C(r)| \le O^*(2^{(1 - h(\rho)) n'} / (k-1)^{r - \rho n'} \prod_i {\lambda_i}^{-\nu_i}  ) $ for all $(r - \rho n') \in [L]^*$.
\end{lemma}
\begin{proof}
    To construct $\mathfrak{C}(L)$ for $H(F, \mathcal{I})$, we construct covering code $C_0(\rho n')$ for $H_0 = \{0,1\}^{n'}$ and $\ell_i$-covering code for $H_i = {A_i}^{\nu_i}$ for all $i$, then take a Cartesian product of all the codes.
    By Lemma~\ref{cover_01_space}, the time taken for constructing $C_0(\rho n')$ is $O^*(2^{(1 - h(\rho)) n'})$, and $|C_0(\rho n')| \le O^*(2^{(1 - h(\rho)) n'})$. By Lemma~\ref{general_space}, for each $i$, the time taken for constructing $C(r_i)$ for each $r_i \in [\ell_i]^*$ is $O^*({\lambda_i}^{- \nu_i} / (k-1)^{r_i})$ and $|C(r_i)| \le O^*({\lambda_i}^{- \nu_i} / (k-1)^{r_i})$.
    So by Lemma~\ref{product_all}, we have that $|C(r)|$ can be upper bounded by:
    \begin{equation*}
        2^{(1 - h(\rho)) n'} \cdot \sum_{\sum_{i} r_i = r - \rho n'} \left( \prod_{i} O^*({\lambda_i}^{-\nu_i} / (k-1)^{r_i}  ) \right) = O^*(2^{(1 - h(\rho)) n'} / (k-1)^{r - \rho n'} \prod_{i} {\lambda_i}^{-\nu_i}  ).
    \end{equation*}
    The equality holds because 
    $L$ is a linear combination of $\nu_i$ with constant coefficients and $\nu_i = \Theta(n)$, thus there are $O(1)$ terms in the product since $\sum_i \nu_i \le n$.
    Meanwhile, there are $O^*(1)$ ways to partition $(r - \rho n')$ into constant number of integers, thus the outer sum has $O^*(1)$ terms.
    Together we get an $O^*(1)$ factor in the right-hand side.

    The construction time includes constructing each covering code for $H_i~(i \ge 0)$ and concatenating each of them by Lemma~\ref{product_all}, which is dominated by the concatenation time. As a result, the time taken to construct $C(r)$ for all $(r - \rho n') \in [L]^*$ is:
    \begin{align*}
        \sum_{(r - \rho n') \in [L]^*} O^*(2^{(1 - h(\rho)) n'} / (k-1)^{r - \rho n'} \prod_{i} {\lambda_i}^{-\nu_i}  )
         = O^*(2^{(1 - h(\rho)) n'} \prod_{i} {\lambda_i}^{-\nu_i} ),
    \end{align*}
    because it is the sum of a geometric series. Therefore conclude the proof.
\end{proof}

Using our generalized covering code and applying Lemma~\ref{full_ball} for \textsf{searchball-fast} (Line~\ref{line_searchball} in Algorithm~\ref{dls_alg}), we can upper bound the running time of \textsf{DLS}.

\begin{lemma}\label{dls_upper_bound}
    Given $k$-CNF $F$ and instance $\mathcal{I}$, \textsf{DLS} runs in time $T_{\text{DLS}} = \mathcal{O}((\frac{2(k-1)}{k})^{n'} \cdot \prod_{i} {\lambda_i}^{-\nu_i})$, where $n' = n - |V(\mathcal{I})|$, $\lambda_i$ is the characteristic value of chain $\mathcal{S}_i$ and $\nu_i$ is number of chains in $\mathcal{I}$ with the same solution space to $\mathcal{S}_i$.
\end{lemma}
\begin{proof}
    The running time includes the construction time for $\mathfrak{C}(L)$ and the total searching time in all Hamming balls. It is easy to show that the total time is dominated by the searching time using Lemma~\ref{all_code}, thus we have the following equation after multiplying a sub-exponential factor $\mathcal{O}(1)$ for the other $o(n)$ chains not in $\mathcal{I}$ (see footnote \ref{footnote}):
    \begin{align*}
        T_{\text{DLS}} &= \mathcal{O}(1) \cdot \sum_{(r - \rho n') \in [L]^*} \left( |C(r)| \cdot (k-1)^{r + o(r)} \right) \\
        & = \mathcal{O}(1) \cdot \sum_{(r - \rho n') \in [L]^*} \left( O^*(2^{(1 - h(\rho)) n'} / (k-1)^{r - \rho n'} \prod_{i} {\lambda_i}^{-\nu_i}  ) \cdot (k-1)^{r + o(r)} \right) \\
        & = \mathcal{O}( 2^{(1 - h(\rho) + \rho \log(k-1))n'} \cdot \prod_{i} {\lambda_i}^{-\nu_i}) \\
        &= \mathcal{O}( (\frac{2(k-1)}{k})^{n'} \cdot \prod_{i} {\lambda_i}^{-\nu_i}) .
    \end{align*}
    The first equality follows from Lemma~\ref{full_ball}, the second inequality is from Lemma~\ref{all_code}, and the last equality follows by setting $\rho = \frac{1}{k}$. Therefore we proved this lemma.
\end{proof}

\section{Upper Bound for k-SAT}\label{UBK_section}

In this section, we give our main result on upper bound for $k$-SAT.

A simple branching algorithm \textsf{BR} for general $k$-SAT is given in Algorithm~\ref{br_k}: Greedily construct a maximal instance $\mathcal{I}$ consisting of independent $1$-chains and branch on all satisfying assignments of it if $|\mathcal{I}|$ is small. \footnote{W.l.o.g., one can negate all negative literals in $\mathcal{I}$ to transform the solution space of $1$-chain to $\{0,1\}^k \backslash 0^k$.}
After fixing all variables in $V(\mathcal{I})$, the remaining formula is a $(k-1)$-CNF due to the maximality of $\mathcal{I}$.
Therefore the running time of \textsf{BR} is at most:
\begin{equation}
    T_{\textsf{BR}} = \mathcal{O}((2^k - 1)^{|\mathcal{I}|} \cdot {c_{k-1}}^{n - k|\mathcal{I}|}), \label{br_k_upper_bound}
\end{equation}
where $\mathcal{O}({c_{k-1}}^n)$ is the worst-case upper bound of a deterministic $(k-1)$-SAT algorithm.

\begin{algorithm}
\caption{Branching Algorithm \textsf{BR} for $k$-SAT}
\label{br_k}
\begin{algorithmic}[1]
\REQUIRE $k$-CNF $F$
\ENSURE a satisfying assignment or \verb"Unsatisfiable" or an instance $\mathcal{I}$
\STATE staring from $\mathcal{I} \leftarrow \emptyset$, \textbf{for} $1$-chain $\mathcal{S}: V(\mathcal{I}) \cap V(\mathcal{S}) = \emptyset$, \textbf{do} $\mathcal{I} \leftarrow \mathcal{I} \cup \mathcal{S}$ 
\IF {$|\mathcal{I}| < \nu n$}
    \FOR {each assignment $\alpha \in \{\{0,1\}^k \backslash 0^k\}^{|\mathcal{I}|}$ of $\mathcal{I}$}
        \STATE solve $F|\alpha$ by deterministic $(k-1)$-SAT algorithm
        \STATE \textbf{return} the satisfying assignment if satisfiable
    \ENDFOR
    \RETURN \verb"Unsatisfiable"
\ELSE
    \RETURN $\mathcal{I}$
\ENDIF
\end{algorithmic}
\end{algorithm}

On the other hand, since there are only $1$-chains in $\mathcal{I}$, by Lemma~\ref{dls_upper_bound} we have:
\begin{equation}
    T_{\text{DLS}} = \mathcal{O}((\frac{2(k-1)}{k})^{n - k |\mathcal{I}|} \cdot \lambda^{-|\mathcal{I}|}) \label{dls_k_upper_bound}.
\end{equation}
It remains to calculate the characteristic value $\lambda$ of $1$-chain $\mathcal{S}^{(k)}$.
We prove the following lemma for the unique solution of linear programming $\text{LP}_A$ in Definition~\ref{charactoristic_def}.

\begin{lemma}\label{1_chain_lp_solution}
    For $1$-chain $\mathcal{S}^{(k)}$, let $A$ be its solution space, then
    the characteristic distribution $\pi$ satisfies
    \begin{equation*}
        \pi(a) = \frac{(k-1)^k}{(2k-2)^k - (k-2)^k} \cdot (1 - (\frac{-1}{k-1})^{d(a, 0^k)}) \textit{~for all~} a \in A ,
    \end{equation*}
    and the characteristic value
    $$\lambda = \frac{k^k}{(2k-2)^k - (k-2)^k}.$$
\end{lemma}
\begin{proof}
    We prove that this is a feasible solution to $\text{LP}_A$.
    Constraint $\pi(a) \ge 0~(\forall a \in A)$ is easy to verify. To show constraint $\sum_{a \in A} \pi(a) = 1$ holds, let $y = d(a, 0^k)$ and note there are $\binom{k}{y}$ different $a \in A$ with $d(a, 0^k) = y$, then multiply $\frac{(2k-2)^k - (k-2)^k}{(k-1)^k}$ on both sides:
    \begin{align*}
        \frac{(2k-2)^k - (k-2)^k}{(k-1)^k} \cdot \sum_{a \in A} \pi(a) &= \sum_{1 \le y \le k} \left( (1 - (\frac{-1}{k-1})^y) \cdot \binom{k}{y} \right) \\
        &= \sum_{0 \le y \le k} \binom{k}{y} - \sum_{0 \le y \le k} \binom{k}{y} (\frac{-1}{k-1})^y \\
        &= 2^k - (\frac{k-2}{k-1})^k \\
        &= \frac{(2k-2)^k - (k-2)^k}{(k-1)^k}.
    \end{align*}
    Thus $\sum_{a \in A} \pi(a) = 1$ holds.

    To prove $\lambda = \sum_{a \in A} \left( \pi(a) \cdot (\frac{1}{k-1}) ^ {d(a, a^*)} \right)$, similar to the previous case, we multiply $\frac{(2k-2)^k - (k-2)^k}{(k-1)^k}$ on both sides.
    Note that adding the term at $a = 0^k$ does not change the sum, then for all $a^* \in A$, we have:
    \begin{align*}
        \text{RHS} &= \sum_{a \in A} (1 - (\frac{-1}{k-1})^{d(a, 0^k)}) \cdot (\frac{1}{k-1}) ^ {d(a, a^*)} \\
        &= \sum_{a \in \{0,1\}^k} (1 - (\frac{-1}{k-1})^{d(a, 0^k)}) \cdot (\frac{1}{k-1}) ^ {d(a, a^*)} \\
        &= \sum_{a \in \{0,1\}^k}(\frac{1}{k-1}) ^ {d(a, a^*)} - \sum_{a \in \{0,1\}^k} (-1)^{d(a, 0^k)} (\frac{1}{k-1}) ^ {d(a, 0^k) + d(a, a^*)} .
    \end{align*}
    The first term is equal to $(\frac{k}{k-1})^k = \text{LHS}$.
    To prove the second term is $0$, note that $\exists i \in [k]$ such that some bit $a^*_i = 1$.
    Partition $\{0,1\}^k$ into two sets $S_0 = \{a \in \{0,1\}^k | a_i = 0 \}$ and $S_1 = \{a \in \{0,1\}^k | a_i = 1 \}$.
    We have the following bijection: For each $a \in S_0$, negate the $i$-th bit to get $a' \in S_1$.
    Then $d(a, 0^k) + d(a, a^*) = d(a', 0^k) + d(a', a^*)$ and $(-1)^{d(a, 0^k)} = - (-1)^{d(a', 0^k)}$, so the sum is $0$.
    Therefore we verified the constraint and proved the lemma.
\end{proof}


Observe from (\ref{br_k_upper_bound}) and (\ref{dls_k_upper_bound}) that $T_{\textsf{BR}}$ is an increasing function of $|\mathcal{I}|$, while $T_{\textsf{DLS}}$ is a decreasing function of it, so $T_{\textsf{BR}} = T_{\textsf{DLS}}$ gives the worst-case upper bound for $k$-SAT.
We solve this equation by plugging in $\lambda$ from Lemma~\ref{1_chain_lp_solution} to get $\nu n$ as the worst-case $|\mathcal{I}|$, and obtain the following theorem as our main result on $k$-SAT.

\begin{theorem}\label{main_k_sat_general_form}
    Given $k \ge 3$,
    if there exists a deterministic algorithm for $(k - 1)$-SAT that runs in time $\mathcal{O}({c_{k-1}}^n)$,
    then there exists a deterministic algorithm for $k$-SAT that runs in time $\mathcal{O}({c_k}^n)$, where
    $$c_k = (2^k - 1)^{\nu} \cdot {c_{k-1}}^{1 - k \nu}$$
    and
    \begin{equation*}
        \nu = \frac{\log(2k - 2) - \log{k} - \log{c_{k-1}}} { \log(2^k - 1) - \log(1 - (\frac{k-2}{2k-2})^k) - k \log{c_{k-1}} }.
    \end{equation*}
\end{theorem}

Note that the upper bound for $3$-SAT implied by this theorem is $O(1.33026^n)$, 
but we can do better by applying Theorem~\ref{main2} (presented later) for $c_3 = 3^{\log{\frac{4}{3}} / \log{\frac{64}{21}}} < 1.32793$ to prove all upper bounds for $k$-SAT ($k \ge 4$) in Table~\ref{table_result} of \S{\ref{intro}}.

\section{Upper Bound for 3-SAT}\label{UB3_section}

In this section, we provide a better upper bound for $3$-SAT by a more ingenious branching algorithm.

First of all, we introduce some additional notations in $3$-CNF simplification, then we present our branching algorithm for $3$-SAT from high-level to all its components. Lastly we show how to combine it with the derandomized local search to achieve a tighter upper bound.

\subsection{Additional Notations}

For every clause $C \in F$, if partial assignment $\alpha$ satisfies $C$, then $C$ is removed in $F|\alpha$. Otherwise, the literals in $C$ assigned to $0$ under $\alpha$ are removed from $C$.
If all the literals in $C$ are removed, which means $C$ is unsatisfied under $\alpha$, we replace $C$ by $\bot$ in $F|\alpha$.
Let $G=F|\alpha$,
for every $C \in F$, we use $C^F$ to denote the clause $C$ in $F$ and $C^G \in G$ the new clause derived from $C$ by assigning variables according to $\alpha$.
We use $\mathcal{F}$ to denote the original input $3$-CNF without instantiating any variable, and $C^{\mathcal{F}}$ is called the \emph{original form} of clause $C$.

Let $\textsf{UP}(F)$ be the CNF derived by running \emph{Unit Propagation} on $F$ until there is no $1$-clause in $F$.
Clearly $F$ is satisfiable iff $\textsf{UP}(F)$ is satisfiable, and $\textsf{UP}$ runs in polynomial time \cite{davis1962machine}.

We will also use the set definition of CNF, i.e.,
for a CNF $F = \bigwedge_{i \in [m]} C_i$, it is equivalent to write $F = \{C_i | i \in[m]\}$.
Define $\mathcal{T}(F), \mathcal{B}(F), \mathcal{U}(F)$ as the set of all the $3$-clauses, $2$-clauses and $1$-clauses in $F$ respectively. 
We have that any $3$-CNF $F = \mathcal{T}(F) \cup \mathcal{B}(F) \cup \mathcal{U}(F)$.

\subsection{Branching Algorithm for 3-SAT}

In this subsection, we give our branching algorithm for $3$-SAT (Algorithm~\ref{BR_alg}).
The algorithm is recursive and follows a depth-first search manner:
\begin{itemize}
    \item Stop the recursion when certain conditions are met (Line~\ref{line_condition} and Line~\ref{line_sat}).
    \item Backtrack when the current branch is unsatisfiable (Line~\ref{line_unsat1}, Line~\ref{line_unsat2} and Line~\ref{line_unsat3}).
    \item Branch on all possible satisfying assignments on a clause and recursively call itself (Line~\ref{line_branch}). Return \verb"Unsatisfiable" if all branches return \verb"Unsatisfiable".
    \item Clause sequence $\mathcal{C}$ stores all the branching clauses from root to the current node.
\end{itemize}

It is easy to show that this algorithm is correct
as long as \emph{procedure} $\mathcal{P}$ maintains satisfiability.

\begin{algorithm}
\caption{Branching Algorithm \textsf{BR} for $3$-SAT}
\label{BR_alg}
\begin{algorithmic}[1]
\REQUIRE $3$-CNF $F$, clause sequence $\mathcal{C}$
\ENSURE a satisfying assignment or \verb"Unsatisfiable" or a clause sequence $\mathcal{C}$
\STATE simplify $F$ by \emph{procedure} $\mathcal{P}$ \label{line_simplify}
\IF {$\bot \in F$}
    \RETURN \verb"Unsatisfiable" \label{line_unsat1}
\ELSIF {\emph{condition} $\Phi$ holds} \label{line_condition}
    \STATE stop the recursion, \emph{transform} $\mathcal{C}$ to an instance $\mathcal{I}$ and \textbf{return} $\mathcal{I}$ \label{line_transform}
\ELSIF {$F$ is $2$-CNF}
    \STATE deterministically solve $F$ in polynomial time
    \IF {$F$ is satisfiable} \label{line_sat}
        \STATE stop the recursion and \textbf{return} the satisfying assignment
    \ELSE
        \RETURN \verb"Unsatisfiable" \label{line_unsat2}
    \ENDIF
\ELSE
    \STATE choose a clause $C$ according to \emph{rule} $\Upsilon$ \label{line_rule}
    \STATE for every satisfying assignment $\alpha_C$ of $C$, call \textsf{BR}$(F|\alpha_C, \mathcal{C} \cup C^{\mathcal{F}})$ \label{line_branch}
    \RETURN \verb"Unsatisfiable" \label{line_unsat3}
\ENDIF
\end{algorithmic}
\end{algorithm}

In what follows, we introduce
(\romannumeral1) the \emph{procedure} $\mathcal{P}$ for simplification (Line~\ref{line_simplify});
(\romannumeral2) the clause choosing \emph{rule} $\Upsilon$ (Line~\ref{line_rule});
(\romannumeral3) the \emph{transformation} from clause sequence to instance (Line~\ref{line_transform});
(\romannumeral4) the termination \emph{condition} $\Phi$ (Line~\ref{line_condition}).
All of them are devoted to analyzing the running time of \textsf{BR} as a function of an instance.

\subsubsection{Simplification Procedure}\label{SP}

The simplification relies on the following two lemmas.

\begin{lemma}[\cite{monien1985solving}]\label{autark_lem}
    Given $3$-CNF $F$ and partial assignment $\alpha$, define $$\mathcal{TB}(F, \alpha)=\{C | C \in \mathcal{B}(\textsf{UP}(F | \alpha), C^F \in \mathcal{T}(F)\}.$$
    If $\bot \notin \textsf{UP}(F | \alpha)$ and $\mathcal{TB}(F, \alpha) = \emptyset$, then $F$ is satisfiable iff $\textsf{UP}(F | \alpha)$ is satisfiable and $\alpha$ is called an \emph{autark}.
\end{lemma}
\begin{proof}
    Recall that $\textsf{UP}$ maintains satisfiability. Let $G = \textsf{UP}(F | \alpha)$.
    If $G$ is satisfiable, then $F$ is obviously satisfiable. Also observe that $G$ is a subset of $F$ since there is neither $1$-clause nor new $2$-clause in $G$, so any satisfying assignment of $F$ satisfies $G$ too.
\end{proof}
We also provide the following stronger lemma to further reduce the formula size.
\begin{lemma}\label{simplification2}
    Given $3$-CNF $F$ and $(l_1 \vee l_2) \in \mathcal{B}(F)$,
    if $\exists C \in \mathcal{TB}(F, l_1 = 1)$ such that $l_2 \in C$,
    then $F$ is satisfiable iff $F\backslash C^F \cup C$ is satisfiable.
\end{lemma}
\begin{proof}
    Clearly $F$ is satisfiable if $F\backslash C^F \cup C$ is.
    Suppose $C = l_2 \vee l_3$ and let $\alpha$ be a satisfying assignment of $F$.
    If $\alpha(l_1) = 1$, then $\textsf{UP}(F | l_1 = 1)$ is satisfiable, thus $F\backslash C^F \cup C$ is also satisfiable since $C \in \textsf{UP}(F | l_1 = 1)$.
    Else if $\alpha(l_1) = 0$, then $\alpha(l_2) = 1$ due to $l_1 \vee l_2$, so $\alpha$ satisfies $C$ and the conclusion follows.
\end{proof}

As a result, $3$-CNF $F$ can be simplified by the following polynomial-time \emph{procedure} $\mathcal{P}$:
for every $(l_1 \vee l_2) \in \mathcal{B}(F)$, if $l_1 = 1$ or $l_2 = 1$ is an autark, then apply Lemma~\ref{autark_lem} to simplify $F$; else apply Lemma~\ref{simplification2} to simplify $F$ if possible.


\begin{lemma}\label{simplify_lem}
    After running $\mathcal{P}$ on $3$-CNF $F$, for any $(l_1 \vee l_2) \in \mathcal{B}(F)$ and for any $2$-clause $C \in \mathcal{TB}(F, l_1 = 1)$, it must be $l_2 \notin C$. This also holds when switching $l_1$ and $l_2$.
\end{lemma}
\begin{proof}
    If $\mathcal{TB}(F, l_1 = 1) = \emptyset$, then $l_1 = 1$ is an autark and $F$ can be simplified by Lemma~\ref{autark_lem}. If $C \in \mathcal{TB}(F, l_1 = 1)$ and $l_2 \in C$, then $F$ can be simplified by Lemma~\ref{simplification2}.
\end{proof}

\subsubsection{Clause Choosing Rule}\label{CCR}

Now we present our clause choosing \emph{rule} $\Upsilon$.
By Lemma~\ref{autark_lem} we can always begin with branching on a $2$-clause with a cost of factor $2$ in the upper bound: Choose an arbitrary literal in any $3$-clause and branch on its two assignments $\{0, 1\}$. This will result in a new $2$-clause otherwise it is an autark and we fix it and continue to choose another literal.

Now let us show the overlapping cases between the current branching clause to the next branching clause.
Let $C_0$ be the branching clause in the father node where $C_0^{\mathcal{F}} = l_0 \vee l_1 \vee l_2$, and let $F_0$ be the formula in the father node.
The \emph{rule} $\Upsilon$ works as follows:
if $\alpha_{C_0}(l_1) = 1$, choose arbitrary $C_1 \in \mathcal{TB}(F_0, l_1 = 1)$;
else if $\alpha_{C_0}(l_2) = 1$, choose arbitrary $C_1 \in \mathcal{TB}(F_0, l_2 = 1)$.


We only discuss the case $\alpha_{C_0}(l_1) = 1$ due to symmetry. We enumerate all the possible forms of $C_1^{\mathcal{F}}$ by discussing what literal is eliminated followed by whether $l_2$ or $\bar{l}_2$ is contained:
\begin{enumerate}
    \item $C_1^{\mathcal{F}} \backslash C_1 = l_3$. $C_1$ becomes a $2$-clause due to elimination of $l_3$. There are three cases: (\romannumeral1) $C_1 = l_2 \vee l_4$, (\romannumeral2) $C_1 = \bar{l}_2 \vee l_4$ or (\romannumeral3) $C_1 = l_4 \vee l_5$. \label{case1}
    \item $C_1^{\mathcal{F}} \backslash C_1 = \bar{l}_1$. $C_1$ becomes a $2$-clause due to elimination of $\bar{l}_1$. There are three cases: (\romannumeral1) $C_1 = l_2 \vee l_3$, (\romannumeral2) $C_1 = \bar{l}_2 \vee l_3$ or (\romannumeral3) $C_1 = l_3 \vee l_4$. \label{case2}
    \item $C_1^{\mathcal{F}} \backslash C_1 = l_2$. This means $l_1 = 1 \Rightarrow l_2 = 0$, and $\alpha_{C_0}(l_1 l_2) = 11$ can be excluded. \label{case3}
    \item $C_1^{\mathcal{F}} \backslash C_1 = \bar{l}_2$. This means $l_1 = 1 \Rightarrow l_2 = 1$, and $\alpha_{C_0}(l_1 l_2) = 10$ can be excluded. \label{case4}
\end{enumerate}

Both Case~\ref{case1}.(\romannumeral1) and Case~\ref{case2}.(\romannumeral1) are impossible due to Lemma~\ref{simplify_lem}.
To sum up, we immediately have the following by merging similar cases with branch number bounded from above:
\begin{itemize}
    \item Case~\ref{case1}.(\romannumeral3): it takes at most $3$ branches in the father node to get $l_3 \vee l_4 \vee l_5$.
    \item Case~\ref{case1}.(\romannumeral2), Case~\ref{case2}.(\romannumeral3) and Case~\ref{case4}: it takes at most $3$ branches in the father node to get $\bar{l}_1 \vee l_3 \vee l_4$ or $\bar{l}_2 \vee l_3 \vee l_4$.
    \item Case~\ref{case3}: it takes at most $2$ branches in the father node to get $l_2 \vee l_3 \vee l_4$.
    \item Case~\ref{case2}.(\romannumeral2): it takes at most $3$ branches in the father node to get $\bar{l}_1 \vee \bar{l}_2 \vee l_3$.
\end{itemize}

To fit \emph{rule} $\Upsilon$, there must be at least one literal assigned to $1$ in the branching clause.
Except Case~\ref{case2}.(\romannumeral2), we get a $2$-clause $C_1$,
and \emph{rule} $\Upsilon$ still applies.

Now consider the case $C_1^{\mathcal{F}} = \bar{l}_1 \vee \bar{l}_2 \vee l_3$. If $\alpha(l_1 l_2) = 11$, we have $C_1^F = l_3$, otherwise we have $C_1^F = 1 \vee l_3$. In other words, the assignment satisfying $C_0 \wedge C_1$ should be $\alpha(l_1 l_2 l_3) \in \{ 010, 100, 011, 101, 111 \}$. Note that $\alpha(l_3) = 0$ in the first two assignments, which does not fit \emph{rule} $\Upsilon$. In this case, we do the following:
Choose an arbitrary literal in any $3$-clause and branch on its two assignments $\{0, 1\}$. Continue this process we will eventually get a new $2$-clause (Lemma~\ref{autark_lem}). Now the first two assignments $\alpha(l_1 l_2 l_3) \in \{ 010, 100\}$ has $4$ branches because of the new branched literal, and we have that all $7$ branches fit \emph{rule} $\Upsilon$ because either $l_3 = 1$ or there is a new $2$-clause. Our key observation is the following:
These $7$ branches correspond to all satisfying assignments of $C_0 \wedge C_1$, which can be amortized to think that $C_1$ has $3$ branches and $C_0$ has $7/3$ branches.
\footnote{\label{note1}Without which or Lemma~\ref{simplify_lem} would ruin our worst-case upper bound, see Appendix~\ref{DD_section}.}
As a conclusion, we modify the last case to be:
\begin{itemize}
    \item Case~\ref{case2}.(\romannumeral2): it takes at most $7/3$ branches in the father node to get $\bar{l}_1 \vee \bar{l}_2 \vee l_3$.
\end{itemize}

\subsubsection{Transformation from Clause Sequence to Instance}\label{cs_to_i}

We show how to transform a clause sequence $\mathcal{C}$ to an instance, then take a symbolic detour to better formalize the cost of generating chains, i.e., the running time of \textsf{BR}.

Similar to above, let $C_1$ be the clause chosen by \emph{rule} $\Upsilon$ and let $C_0$ be the branching clause in the father node, moreover let $C$ be the branching clause in the grandfather node. In other words, $C_1, C_0, C$ are the last three clauses in $\mathcal{C}$. $C_1$ used to be a $3$-clause in the father node since $C_1 \in \mathcal{T}(F)$, thus $C_1$ is independent with $C$ because all literals in $C$ are assigned to some values in $F$, so $C_1$ can only overlap with $C_0$. Therefore, clauses in $\mathcal{C}$ can only (but not necessarily) overlap with the clauses next to them.

By the case discussion in \S{\ref{CCR}}, there are only $4$ overlapping cases between $C_0$ and $C_1$, which we call \emph{independent} for $\langle l_0 \vee l_1 \vee l_2, ~l_3 \vee l_4 \vee l_5 \rangle$, \emph{negative} for $\langle l_0 \vee l_1 \vee l_2, ~\bar{l}_1 \vee l_3 \vee l_4 \rangle$ or $\langle l_0 \vee l_1 \vee l_2, ~\bar{l}_2 \vee l_3 \vee l_4 \rangle$, \emph{positive} for $\langle l_0 \vee l_1 \vee l_2, ~l_2 \vee l_3 \vee l_4 \rangle$ and \emph{two-negative} for $\langle l_0 \vee l_1 \vee l_2, ~\bar{l}_1 \vee \bar{l}_2 \vee l_3 \rangle$.
There is a natural mapping from clause sequence to a string.

\begin{definition}\label{to_string}
    Let $\mathcal{C}$ be a clause sequence,
    define function $\zeta: \mathcal{C} \mapsto \Gamma^{|\mathcal{C}|}$, where $\Gamma = \{\verb"*", \verb"n", \verb"p", \verb"t"\}$, satisfies that the $i$-th bit of $\zeta(\mathcal{C})$ is \verb"*" if $\mathcal{C}_i$ and $\mathcal{C}_{i+1}$ are independent, or \verb"n" if negative, or \verb"p" if positive, or \verb"t" if two-negative for all $i \in [|\mathcal{C}| - 1]$, and the $|\mathcal{C}|$-th bit of $\zeta(\mathcal{C})$ is \verb"*".
    A $\tau$-chain $\mathcal{S}$ is also a clause sequence of length $\tau$, so $\zeta$ maps $\mathcal{S}$ to $\Gamma^{\tau}$.
    Two chains $\mathcal{S}_1$ and $\mathcal{S}_2$ are \emph{isomorphic} if $\zeta(\mathcal{S}_1) = \zeta(\mathcal{S}_2)$.
\end{definition}

Then the \emph{transformation} from $\mathcal{C}$ to $\mathcal{I}$ naturally follows:
Partition $\zeta(\mathcal{C})$ by \verb"*", then every substring corresponds to a chain, just add this chain to $\mathcal{I}$.
Now we can formalize the cost.

\begin{lemma}\label{branch_upper_bound_1}
    Given $3$-CNF $\mathcal{F}$, let $\mathcal{C}$ be the clause sequence in time $T$ of running \textsf{BR}$(\mathcal{F}, \emptyset)$, it must be $T \le O^*(2^{\kappa_1} \cdot 3^{\kappa_2} \cdot (7/3)^{\kappa_3})$, where $\kappa_1$ is the number of \verb"p" in $\zeta(\mathcal{C})$, $\kappa_2$ is the number of \verb"*" and \verb"n" in $\zeta(\mathcal{C})$, and $\kappa_3$ is the number of \verb"t" in $\zeta(\mathcal{C})$.
\end{lemma}
\begin{proof}
    By Definition~\ref{to_string} and case discussion in \S{\ref{CCR}}, the conclusion follows.
\end{proof}

\subsubsection{Termination Condition}\label{tc_type_def}

We show how the cost of generating chains implies the termination \emph{condition} $\Phi$.
We map every chain to an integer as the \emph{type} of the chain such that isomorphic chains have the same type.
Formally, let $\mathcal{I}(\mathcal{C})$ be the instance transformed from $\mathcal{C}$, and let $\Sigma = \{ \zeta(\mathcal{S}) | \mathcal{S} \in \mathcal{I}(\mathcal{C}) \}$ be the set of distinct strings with no repetition.
Define bijective function $g: \Sigma \mapsto [\theta]$ that maps each string $\zeta(\mathcal{S})$ in $\Sigma$ to a distinct integer as the \emph{type} of chain $\mathcal{S}$,
where $\theta = |\Sigma|$ is the number of types of chain in $\mathcal{C}$ and $g$ can be arbitrary fixed bijection.
Define \emph{branch number} $b_i$ of type-$i$ chain $\mathcal{S}$ as $b_i = 2^{\kappa_1} \cdot 3^{\kappa_2} \cdot (7 / 3)^{\kappa_3}$, where $\kappa_1$ is the number of \verb"p" in $\zeta(\mathcal{S})$, $\kappa_2$ is the number of \verb"*" and \verb"n" in $\zeta(\mathcal{S})$, and $\kappa_3$ is the number of \verb"t" in $\zeta(\mathcal{S})$.
Also define the \emph{chain vector} $\vec{\nu}
\in \mathbb{Z}^{\theta}$ for $\mathcal{I}(\mathcal{C})$ satisfies $\nu_i = \left|\{\mathcal{S} \in \mathcal{I}(\mathcal{C}) | (g \circ \zeta)(\mathcal{S}) = i \}\right|$ for all $i \in [\theta]$, i.e., $\nu_i$ is the number of type-$i$ chains in $\mathcal{I}(\mathcal{C})$. We can rewrite Lemma~\ref{branch_upper_bound_1} as the following.
\begin{corollary}\label{branch_upper_bound_2}
    Given $3$-CNF $\mathcal{F}$, let $\mathcal{I}$ be the instance in time $T$ of running \textsf{BR}$(\mathcal{F}, \emptyset)$, it must be $T \le T_{\textsf{BR}} = O^*(\prod_{i \in [\theta]} b_i^{\nu_i})$, where $b_i$ is the branch number of type-$i$ chain and $\vec{\nu}$ is the chain vector for $\mathcal{I}$.
\end{corollary}

To achieve worst-case upper bound $\mathcal{O}(c^n)$ for solving $3$-SAT, we must have $T_{\textsf{BR}} \le \mathcal{O}(c^n)$, which is $\prod_{i=1}^{\theta} b_i^{\nu_i} \le c^n$. This immediately gives us the termination \emph{condition} $\Phi$:
$(\sum_{i \in [\theta]} \nu_i \cdot \log b_i) / \log c > n$.


Therefore, we can hardwire such condition into the algorithm to achieve the desired upper bound, as calculated in the next subsection.

\subsection{Combination of Two Algorithms}\label{comb}

By combining \textsf{BR} and \textsf{DLS} as in Algorithm~\ref{Framework}, we have that the worst-case upper bound $\mathcal{O}(c^n)$ is attained when $T_{\textsf{BR}} = T_{\textsf{DLS}}$, which is:
\begin{equation}
    c^n  = \prod_{i \in [\theta]} b_i^{\nu_i} = (\frac{4}{3})^{n'} \cdot \prod_{i \in [\theta]} {\lambda_i}^{-\nu_i} , \label{equal1}
\end{equation}
followed by Corollary~\ref{branch_upper_bound_2} and Lemma~\ref{dls_upper_bound}.
Let $\eta_i$ be the number of variables in a type-$i$ chain for all $i \in [\theta]$, we have that $n' = n - |V(\mathcal{I})| = n - \sum_{i \in [\theta]} \eta_i \nu_i$. Taking the logarithm and divided by $n$, (\ref{equal1}) becomes:
\begin{equation}
    \log c = \sum_{i \in [\theta]} \frac{\nu_i}{n} \log{b_i} = \log{\frac{4}{3}} - \sum_{i \in [\theta]} \frac{\nu_i}{n} (\eta_i \log{\frac{4}{3}} + \log{\lambda_i}) \label{equal2} .
\end{equation}
The second equation is a linear constraint over $\frac{1}{n} \cdot \vec{\nu}$, which gives that $\log c$ is maximized when $\nu_i = 0$ for all $i \neq \arg \max_{i \in [\theta]} \{ \log{b_i}  / (\log b_i + \eta_i \log{\frac{4}{3}} + \log{\lambda_i} ) \}$.

Based on the calculation of $\text{LP}_A$ (see Appendix \ref{3sat_values}), we show that chain $\mathcal{S}$ with $\zeta(\mathcal{S}) = \verb"*"$ (say, type-$1$ chain) corresponds to the maximum value above, namely:
\begin{equation*}
    \arg \max_{i \in [\theta]} \{ \log{b_i}  / (\log b_i + \eta_i \log{\frac{4}{3}} + \log{\lambda_i} ) \} = 1.
\end{equation*}
In other words, all chains in $\mathcal{I}$ are $1$-chain.
Substitute $\lambda_1 = 3/7, b_1 = 3, \eta_1 = 3$ and $\nu_i = 0$ for all $i \in [2, \theta]$ into (\ref{equal2}) (see Table~\ref{3sat_numerical} in Appendix~\ref{3sat_values}), we obtain our main result on $3$-SAT as follow.
\begin{theorem}\label{main2}
    There exists a deterministic algorithm for $3$-SAT that runs in time $\mathcal{O}(3^{n \log{\frac{4}{3}} / \log{\frac{64}{21}}})$.
\end{theorem}

This immediately implies the upper bound $O(1.32793^n)$ for $3$-SAT in Table~\ref{table_result} of \S{\ref{intro}}.

\section{Conclusion and Discussion}\label{conclusion}

We have shown that how to improve Moser and Scheder's deterministic $k$-SAT algorithm by combining with a branching algorithm.
Specifically, for $3$-SAT we design a novel branching algorithm which reduces the branch number of $1$-chain from $7$ to $3$.
In general, we expect to see $1$-chain in $k$-CNF to have branch number $2^{k-1} - 1$ instead of $2^k - 1$, therefore improving the upper bound for Algorithm~\ref{br_k} from $\mathcal{O}((2^k - 1)^{|\mathcal{I}|} \cdot {c_{k-1}}^{n - k|\mathcal{I}|})$
to
$\mathcal{O}((2^{k-1} - 1)^{|\mathcal{I}|} \cdot {c_{k-1}}^{n - k|\mathcal{I}|})$
as for Algorithm~\ref{BR_alg}.
However, this requires much more work using the techniques developed in this paper.

We believe that there exists an elegant proof for the analysis of branching algorithm on $k$-SAT instead of tedious case analysis, and it is tight under the current framework, i.e., the combination of a branching algorithm and the derandomized local search, leveraged by chain.

In a recent work, the technique in this paper is generalized to give an improved deterministic algorithm for NAE-$k$-SAT \cite{liu2018curse}, which achieves upper bound that is better than $k$-SAT algorithms for the first time.

\paragraph{Acknowledgements.}{The author wants to thank Yuping Luo, S. Matthew Weinberg and Periklis A. Papa-konstantinou for helpful discussions.
Research at Princeton University partially supported by an innovation research grant from Princeton and a gift from Microsoft.}


\bibliographystyle{alpha}
\bibliography{det_ksat}

\newcommand{\etalchar}[1]{$^{#1}$}
\begin{thebibliography}{DGH{\etalchar{+}}02}

\bibitem[BK04]{DBLP:journals/tcs/BrueggemannK04}
Tobias Br{\"{u}}ggemann and Walter Kern.
\newblock An improved deterministic local search algorithm for 3-sat.
\newblock {\em Theoretical Computer Science}, 329(1-3):303--313, 2004.

\bibitem[DGH{\etalchar{+}}02]{dantsin2002deterministic}
Evgeny Dantsin, Andreas Goerdt, Edward~A Hirsch, Ravi Kannan, Jon Kleinberg,
  Christos Papadimitriou, Prabhakar Raghavan, and Uwe Sch{\"o}ning.
\newblock A deterministic (2-2/(k+1))\({}^{\mbox{n}}\) algorithm for k-sat
  based on local search.
\newblock {\em Theoretical Computer Science}, 289(1):69--83, 2002.

\bibitem[DLL62]{davis1962machine}
Martin Davis, George Logemann, and Donald Loveland.
\newblock A machine program for theorem-proving.
\newblock {\em Communications of the ACM}, 5(7):394--397, 1962.

\bibitem[Her14]{hertli20143}
Timon Hertli.
\newblock 3-sat faster and simpler---unique-sat bounds for ppsz hold in
  general.
\newblock {\em SIAM Journal on Computing}, 43(2):718--729, 2014.

\bibitem[HSSW02]{hofmeister2002probabilistic}
Thomas Hofmeister, Uwe Sch{\"o}ning, Rainer Schuler, and Osamu Watanabe.
\newblock A probabilistic 3-sat algorithm further improved.
\newblock In {\em 19th Annual Symposium on Theoretical Aspects of Computer
  Science, STACS 2002}, pages 192--202. Springer, 2002.

\bibitem[IT04]{iwama2004improved}
Kazuo Iwama and Suguru Tamaki.
\newblock Improved upper bounds for 3-sat.
\newblock In {\em Proceedings of the Fifteenth Annual ACM-SIAM Symposium on
  Discrete Algorithms, SODA 2004}, volume~4, pages 328--328, 2004.

\bibitem[KK09]{DBLP:series/faia/BuningK09}
Hans Kleine{ }B{\"{u}}ning and Oliver Kullmann.
\newblock Minimal unsatisfiability and autarkies.
\newblock In {\em Handbook of Satisfiability}, pages 339--401. 2009.

\bibitem[Kul99]{DBLP:journals/tcs/Kullmann99}
Oliver Kullmann.
\newblock New methods for 3-sat decision and worst-case analysis.
\newblock {\em Theoretical Computer Science}, 223(1-2):1--72, 1999.

\bibitem[Liu18a]{liu2018ksat}
S.~Cliff Liu.
\newblock Chain, generalization of covering code, and deterministic algorithm
  for k-sat.
\newblock In {\em 45th International Colloquium on Automata, Languages, and
  Programming, {ICALP} 2018, July 9-13, 2018}, pages 88:1--88:13, 2018.

\bibitem[Liu18b]{liu2018curse}
S~Cliff Liu.
\newblock The curse and blessing of not-all-equal in k-satisfiability.
\newblock {\em CoRR}, abs/1809.04312, 2018.

\bibitem[Liu20]{DBLP:journals/corr/abs-2001-06536}
S.~Cliff Liu.
\newblock Simpler partial derandomization of {PPSZ} for k-sat.
\newblock {\em CoRR}, abs/2001.06536, 2020.

\bibitem[MS85]{monien1985solving}
Burkhard Monien and Ewald Speckenmeyer.
\newblock Solving satisfiability in less than $2^n$ steps.
\newblock {\em Discrete Applied Mathematics}, 10(3):287--295, 1985.

\bibitem[MS11]{moser2011full}
Robin~A. Moser and Dominik Scheder.
\newblock A full derandomization of sch{\"o}ning's k-sat algorithm.
\newblock In {\em Proceedings of the Forty-third Annual ACM Symposium on Theory
  of Computing, {STOC} 2011}, pages 245--252, 2011.

\bibitem[MTY13]{DBLP:journals/algorithmica/MakinoTY13}
Kazuhisa Makino, Suguru Tamaki, and Masaki Yamamoto.
\newblock Derandomizing the {HSSW} algorithm for 3-sat.
\newblock {\em Algorithmica}, 67(2):112--124, 2013.

\bibitem[PPSZ05]{DBLP:journals/jacm/PaturiPSZ05}
Ramamohan Paturi, Pavel Pudl{\'{a}}k, Michael~E. Saks, and Francis Zane.
\newblock An improved exponential-time algorithm for \emph{k}-sat.
\newblock {\em J. {ACM}}, 52(3):337--364, 2005.

\bibitem[Rol05]{DBLP:conf/sat/Rolf05}
Daniel Rolf.
\newblock Derandomization of {PPSZ} for unique- \emph{k}-sat.
\newblock In {\em Theory and Applications of Satisfiability Testing, 8th
  International Conference, {SAT} 2005}, pages 216--225, 2005.

\bibitem[Sch70]{Schiermeyer1970Solving}
Ingo Schiermeyer.
\newblock Solving 3-satisfiability in less than $1.579^n$ steps.
\newblock {\em Computer Science Logic}, pages 379--394, 1970.

\bibitem[Sch99]{schoning1999probabilistic}
Uwe Sch{\"{o}}ning.
\newblock A probabilistic algorithm for k-sat and constraint satisfaction
  problems.
\newblock In {\em 40th Annual Symposium on Foundations of Computer Science,
  {FOCS} 1999}, pages 410--414, 1999.

\bibitem[Sch08]{DBLP:conf/latin/Scheder08}
Dominik Scheder.
\newblock Guided search and a faster deterministic algorithm for 3-sat.
\newblock In {\em the 3rd Latin American Theoretical Informatics Symposium,
  LATIN 2008}, pages 60--71, 2008.

\end{thebibliography}

\appendix

\section{Generation of All Types of Chain for 3-SAT}\label{3sat_values}

In \S{\ref{cs_to_i}}, we proved that there are only $4$ overlapping cases between successive clauses, thus for any $\mathcal{S} \in \mathcal{I}$, $\zeta(\mathcal{S}) \in \{\verb"n", \verb"p", \verb"t"\}^* \uplus \{\verb"*"\}$.
Now we show that any $\zeta(\mathcal{C})$ cannot have substring \verb"tp" or \verb"tt", which greatly reduces the number of types of chain. Recall that in \S{\ref{CCR}}, if $C_1^{\mathcal{F}} = \bar{l}_1 \vee \bar{l}_2 \vee l_3$, then $C_0 \wedge C_1$ has at most $7$ branches, where $4$ of them correspond to fixing all variables in $C_1$ and branching on a new literal. These necessarily lead to a branching clause independent with $C_1^{\mathcal{F}}$.
Also note that when $\alpha(l_3) = 1$ for the remaining $3$ branches, the next branching clause cannot have $l_3$, otherwise it is eliminated.
As a result, the clause in $\mathcal{C}$ right after $C_1$ can only be independent or negative overlapping with $C_1$, which means \verb"tp" or \verb"tt" is not a substring of $\zeta(\mathcal{C})$.

Then we prove that there are only finite types of chain. It is sufficient to prove that the length of chain is upper bounded by a constant. When choosing clause to branch, we can always choose a literal in some $3$-clause and branch on this literal to get a new $2$-clause (Lemma~\ref{autark_lem}). This costs us a factor of $2$ in the branch number but results in an independent branching clause.
Observe that as long as the new branch number ${b_i}' = 2 b_i$ satisfies $\log{{b_i}'}  / (\log {b_i}' + \eta_i \log{\frac{4}{3}} + \log{\lambda_i} ) \le \log{b_1}  / (\log b_1 + \eta_1 \log{\frac{4}{3}} + \log{\lambda_1} )$, it does not influence the worst-case upper bound.

To sum up, a string $\zeta$ corresponds to a type-$i$ chain can be generated by the following two rules.
The second rule can be applied whenever at our will.

\begin{enumerate}
    \item $\zeta \leftarrow (\zeta \uplus \verb"*")$ or $(\zeta \uplus \verb"n")$ or $(\zeta \uplus \verb"p")$ or $(\zeta \uplus \verb"t*")$ or $(\zeta \uplus \verb"tn")$.
    \item $\zeta \leftarrow (\zeta \uplus \verb"*")$ if $\log{(2{b_i})}  / (\log {(2b_i)} + \eta_i \log{\frac{4}{3}} + \log{\lambda_i} ) \le \log{b_1}  / (\log b_1 + \eta_1 \log{\frac{4}{3}} + \log{\lambda_1} )$.
\end{enumerate}

We report all chains $\mathcal{S}_i$ of type-$i$ with characteristic value $\lambda_i$ and $f_i = \log{b_i}  / (\log b_i + \eta_i \log{\frac{4}{3}} + \log{\lambda_i})$.
The characteristic values are given by solving linear programming $\text{LP}_A$ from Definition~\ref{charactoristic_def}.
The variable number $\eta_i = |V(\mathcal{S}_i)|$, branch number $b_i$ are trivial to calculate, thus do not report here.
Note that the reversed (except the terminal \verb"*") string is equivalent to the original one.
Chain generated by rule 2 is marked with \textbf{r2} in their type.
Using a breath-first search, one can easily check that Table~\ref{3sat_numerical} lists all the possible types of chains in our branching algorithm for $3$-SAT.

\begin{table}
\centering
\caption{The characteristic value $\lambda_i$ and $f_i = \log{b_i}  / (\log b_i + \eta_i \log{\frac{4}{3}} + \log{\lambda_i})$ for all types $i$ of possible chains generated by our branching algorithm for $3$-SAT.}
\label{3sat_numerical}
\begin{minipage}{0.45\textwidth}

\begin{tabular}{|l|l|l|l|}
\hline
type-$i$ & $\zeta(\mathcal{S}_i)$  & $\lambda_i$ & $f_i$ \\\hline

1 & \verb"*" & $3/7$ & $0.98586\dots$ \\\hline

2 & \verb"n*" & $27/110$ & $0.984\dots$  \\\hline

3 & \verb"p*" & $81/331$ & $0.983\dots$ \\\hline

4 & \verb"t*" & $15/46$  & $0.984\dots$ \\\hline

5 & \verb"nn*" & $9/64$  & $0.984\dots$ \\\hline

6 & \verb"np*" & $81/578$  & $0.983\dots$ \\\hline

7 & \verb"nt*" & $45/241$  & $0.984\dots$ \\\hline


8 \textbf{r2}  & \verb"pp*" & $243/1739$  & $0.98580\dots$ \\\hline

9 & \verb"pt*" & $27/145$ & $0.983\dots$  \\\hline


10 & \verb"nnn*" &  $243 / 3016$  & $0.984\dots$ \\\hline

11 & \verb"nnp*" & $729/9080$  & $0.983\dots$ \\\hline

12 & \verb"nnt*" & $135/1262$  & $0.984\dots$ \\\hline

13 & \verb"npn*" & $243/3028$  & $0.983\dots$ \\\hline

14 \textbf{r2}  & \verb"npp*" & $729/9110$  & $0.9853\dots$ \\\hline

15 & \verb"npt*" & $45/422$  & $0.983\dots$ \\\hline


16 & \verb"ntn*" & $405/3788$  & $0.984\dots$ \\\hline

17 \textbf{r2}  & \verb"pnp*" & $2187/27334$  & $0.9853\dots$ \\\hline

18 & \verb"pnt*" & $405/3799$  & $0.983\dots$ \\\hline



19 & \verb"tnt*" & $25/176$  & $0.984\dots$ \\\hline

\end{tabular}
\end{minipage}
\hfil
\begin{minipage}{0.45\textwidth}

\begin{tabular}{|l|l|l|l|}
\hline

type-$i$ & $\zeta(\mathcal{S}_i)$ & $\lambda_i$ & $f_i$ \\\hline

20 \textbf{r2}  & \verb"nnnn*" & $243/5264$  & $0.98583\dots$ \\\hline

21 \textbf{r2}  & \verb"nnnp*" & $729/15848$  & $0.9854\dots$ \\\hline

22 & \verb"nnnt*" & $405/6608$  & $0.984\dots$ \\\hline

23  \textbf{r2} & \verb"nnpn*" & $729/15856$  & $0.9855\dots$ \\\hline

24 \textbf{r2}  & \verb"nnpp*" & $2187/47704$  & $0.984\dots$ \\\hline

25 \textbf{r2}  & \verb"nnpt*" & $1215/19888$  & $0.9854\dots$ \\\hline


26 \textbf{r2}  & \verb"npnp*" & $2187/47732$  & $0.9850\dots$ \\\hline

27 \textbf{r2}  & \verb"npnt*" & $405/6634$  & $0.9856\dots$ \\\hline

28 & \verb"ntnn*" & $135/2204$  & $0.984\dots$ \\\hline

29 \textbf{r2}  & \verb"ntnp*" & $1215/19904$  & $0.9856\dots$ \\\hline

30 & \verb"ntnt*" & $675/8299$  & $0.984\dots$ \\\hline


31 \textbf{r2}  & \verb"pnnp*" & $729/15904$  & $0.9850\dots$ \\\hline

32 \textbf{r2}  & \verb"pnnt*" & $1215/19894$  & $0.9855\dots$ \\\hline



33 & \verb"tnnt*" & $45/553$  & $0.984\dots$ \\\hline


34 \textbf{r2}  & \verb"tnpp*" & $405/6653$  & $0.9850\dots$ \\\hline

35 \textbf{r2}  & \verb"tnpt*" & $675/8321$ & $0.9855\dots$  \\\hline

36 \textbf{r2} & \verb"tnnnn*" & $243/6920$  & $0.9855\dots$ \\\hline

37 \textbf{r2}  & \verb"tnnnp*" & $3645/104168$  & $0.9852\dots$ \\\hline

38 \textbf{r2}  & \verb"tnnnt*" & $225/4826$  & $0.9856\dots$ \\\hline

\end{tabular}
\end{minipage}
\end{table}

Consequently, $\arg \max_{i \in [38]} \{f_i\}  = 1$ and $\mathcal{I}$ consisting of only $1$-chains is indeed our worst case.

\section{Degeneration of Algorithm}\label{DD_section}

As a simple degeneration, we illustrate why the algorithm yields a worse upper bound without Lemma~\ref{simplify_lem} or the key observation made in \S{\ref{CCR}}.

Recall in the case study of \S{\ref{CCR}}, if we do not apply Lemma~\ref{simplify_lem}, then Case~\ref{case1}.(\romannumeral1) and Case~\ref{case2}.(\romannumeral1) are possible, which give branch number $b = 9$ for the positive overlapping $2$-chain. Suppose the instance returned by \textsf{BR} contains only $\nu$ such $2$-chains, then in the worst case, we have $c^n = 9^{\nu} = (\frac{4}{3})^{n - 5\nu} \cdot \lambda^{-\nu} \ge 1.328^n$ (see type-$3$ in Table~\ref{3sat_numerical} for $\lambda = 81/331$).

Similarly, without the amortized analysis in \S{\ref{CCR}}, a two-negative overlapping $2$-chain would have branch number $9$, which gives $c^n = 9^{\nu} = (\frac{4}{3})^{n - 4\nu} \cdot \lambda^{-\nu} \ge 1.328^n$ in the worst case (see type-$4$ in Table~\ref{3sat_numerical} for $\lambda = 15/46$).

Therefore our optimizations are necessary for proving Theorem~\ref{main2}.

\end{document}